\documentclass{article}

\def\noheaderplainsetup{

\topmargin=0pt \headheight=0pt \headsep=0pt  \oddsidemargin=0pt \evensidemargin=0pt  \textheight=9.1truein \textwidth=6.5truein}   

\noheaderplainsetup

\usepackage{amsfonts}

\begin{document}


\newcommand{\cltw}{\mbox{LKg$^0$}}


\newcommand{\sqc}{\mbox{\small \raisebox{0.0cm}{$\bigtriangleup$}}} 

\newcommand{\sqd}{\mbox{\small \raisebox{0.049cm}{$\bigtriangledown$}}} 

\newcommand{\code}[1]{\ulcorner #1 \urcorner}
\newcommand{\pintimpl}{\mbox{\hspace{2pt}\raisebox{0.033cm}{\tiny $>$}\hspace{-0.18cm} \raisebox{-0.043cm}{\large --}\hspace{2pt}}} 
\newcommand{\st}{\mbox{\raisebox{-0.05cm}{$\circ$}\hspace{-0.13cm}\raisebox{0.16cm}{\tiny $\mid$}\hspace{2pt}}}  
\newcommand{\sti}{\mbox{\raisebox{-0.02cm}
{\scriptsize $\circ$}\hspace{-0.121cm}\raisebox{0.08cm}{\tiny $.$}\hspace{-0.079cm}\raisebox{0.10cm}
{\tiny $.$}\hspace{-0.079cm}\raisebox{0.12cm}{\tiny $.$}\hspace{-0.085cm}\raisebox{0.14cm}
{\tiny $.$}\hspace{-0.079cm}\raisebox{0.16cm}{\tiny $.$}\hspace{1pt}}}
\newcommand{\intimpl}{\mbox{\hspace{2pt}$\circ$\hspace{-0.14cm} \raisebox{-0.043cm}{\Large --}\hspace{2pt}}} 
\newcommand{\fintimpl}{\mbox{\hspace{2pt}$\bullet$\hspace{-0.14cm} \raisebox{-0.058cm}{\Large --}\hspace{-6pt}\raisebox{0.008cm}{\scriptsize $\wr$}\hspace{-1pt}\raisebox{0.008cm}{\scriptsize $\wr$}\hspace{4pt}}} 

\newcommand{\adi}{\hspace{2pt}\raisebox{0.02cm}{\mbox{\small $\sqsupset$}}\hspace{2pt}} 
\newcommand{\plus}{\mbox{\hspace{1pt}\raisebox{0.05cm}{\tiny\boldmath $+$}\hspace{1pt}}}
\newcommand{\minus}{\mbox{\hspace{1pt}\raisebox{0.05cm}{\tiny\boldmath $-$}\hspace{1pt}}}
\newcommand{\mult}{\mbox{\hspace{1pt}\raisebox{0.05cm}{\tiny\boldmath $\times$}\hspace{1pt}}}
\newcommand{\equals}{\mbox{\hspace{1pt}\raisebox{0.05cm}{\tiny\boldmath $=$}\hspace{1pt}}}
\newcommand{\notequals}{\mbox{\hspace{1pt}\raisebox{0.05cm}{\tiny\boldmath $\not=$}\hspace{1pt}}}
\newcommand{\successor}{\mbox{\hspace{1pt}\boldmath $'$}}

\newcommand{\elz}[1]{\mbox{$\parallel\hspace{-3pt} #1 \hspace{-3pt}\parallel$}} 
\newcommand{\elzi}[1]{\mbox{\scriptsize $\parallel\hspace{-3pt} #1 \hspace{-3pt}\parallel$}}
\newcommand{\emptyrun}{\langle\rangle} 
\newcommand{\oo}{\bot}            
\newcommand{\pp}{\top}            
\newcommand{\xx}{\wp}               
\newcommand{\legal}[2]{\mbox{\bf Lr}^{#1}_{#2}} 
\newcommand{\win}[2]{\mbox{\bf Wn}^{#1}_{#2}} 
\newcommand{\watom}[2]{#1\{#2\}} 
\newcommand{\seq}[1]{\langle #1 \rangle}           


\newcommand{\pst}{\mbox{\raisebox{-0.01cm}{\scriptsize $\wedge$}\hspace{-4pt}\raisebox{0.16cm}{\tiny $\mid$}\hspace{2pt}}}
\newcommand{\pcost}{\mbox{\raisebox{0.12cm}{\scriptsize $\vee$}\hspace{-4pt}\raisebox{0.02cm}{\tiny $\mid$}\hspace{2pt}}}

\newcommand{\gneg}{\mbox{\small $\neg$}}                  
\newcommand{\mli}{\hspace{2pt}\mbox{\small $\rightarrow$}\hspace{2pt}}                      
\newcommand{\cla}{\mbox{$\forall$}}      
\newcommand{\cle}{\mbox{$\exists$}}        
\newcommand{\mld}{\hspace{2pt}\mbox{\small $\vee$}\hspace{2pt}}     
\newcommand{\mlc}{\hspace{2pt}\mbox{\small $\wedge$}\hspace{2pt}}   
\newcommand{\mlci}{\hspace{2pt}\mbox{\footnotesize $\wedge$}\hspace{2pt}}   
\newcommand{\ade}{\mbox{\large $\sqcup$}}      
\newcommand{\ada}{\mbox{\large $\sqcap$}}      
\newcommand{\add}{\hspace{2pt}\mbox{\small $\sqcup$}\hspace{2pt}}                     
\newcommand{\adc}{\hspace{2pt}\mbox{\small $\sqcap$}\hspace{2pt}} 
\newcommand{\adci}{\hspace{2pt}\mbox{\footnotesize $\sqcap$}\hspace{2pt}}              
\newcommand{\clai}{\forall}     
\newcommand{\clei}{\exists}        
\newcommand{\tlg}{\bot}               
\newcommand{\twg}{\top}               
\newcommand{\col}[1]{\mbox{$#1$:}}


\newtheorem{theoremm}{Theorem}[section]
\newtheorem{factt}[theoremm]{Fact}
\newtheorem{corollaryy}[theoremm]{Corollary}
\newtheorem{definitionn}[theoremm]{Definition}
\newtheorem{thesiss}[theoremm]{Thesis}
\newtheorem{lemmaa}[theoremm]{Lemma}
\newtheorem{conventionn}[theoremm]{Convention}
\newtheorem{examplee}[theoremm]{Example}
\newtheorem{exercisee}[theoremm]{Exercise}
\newtheorem{remarkk}[theoremm]{Remark}
\newenvironment{definition}{\begin{definitionn} \em}{ \end{definitionn}}
\newenvironment{thesis}{\begin{thesiss} \em}{ \end{thesiss}}
\newenvironment{theorem}{\begin{theoremm}}{\end{theoremm}}
\newenvironment{lemma}{\begin{lemmaa}}{\end{lemmaa}}
\newenvironment{fact}{\begin{factt}}{\end{factt}}
\newenvironment{corollary}{\begin{corollaryy}}{\end{corollaryy}}
\newenvironment{convention}{\begin{conventionn} \em}{\end{conventionn}}
\newenvironment{example}{\begin{examplee} \em}{\end{examplee}}
\newenvironment{exercise}{\begin{exercisee} \em}{\end{exercisee}}
\newenvironment{remark}{\begin{remarkk} \em}{\end{remarkk}}
\newenvironment{proof}{ {\bf Proof.} }{\  \rule{2.5mm}{2.5mm} \vspace{.2in} }

\newcommand{\nso}{nongshim}
\newcommand{\ns}[3]{\mbox{\it nongshim(#1,#2,#3)}}

\title{A Heuristic Proof Procedure for Propositional Logic}
\author{Keehang Kwon\\ 
  Department of AI, DongA University, \\
   khkwon@dau.ac.kr}
\date{}
\maketitle

Theorem proving  is one of the oldest applications which require heuristics
to prune the search space.  Invertible  proof procedures has been the major
tool.
In this paper, we present a novel and powerful  heuristic called  $\nso$ which can be
seen as an underlying principle of invertible proof procedures.
 Using this  heuristic, we  derive an
 invertible sequent calculus\cite{Ketonen,Troe}
from sequent calculus for propositional logic.
\section{Introduction}\label{intr}

Theorem proving  is one of the oldest applications which require heuristics
to prune the search space.  One key observation about  proof search is that some rules
are {\it invertible}, that is, the premisses are derivable whenever the conclusion is
derivable. We can apply invertible rules  whenever possible without
losing completeness. 
The proof search strategy that first applies all
invertible rules is called {\it inversion}. 
An inversion calculus  has been developed 
for classical, intuitionistic and linear logic.  However, it has the following
drawbacks:

\begin{itemize}

\item Whenever  a new logic  comes along,
 its corresponding inversion calculus must be separately developed and proved to
 be correct. This is a time consuming and redundant process.
 
 \item Inversion calculus processes connectives eagerly, often processing them
  even when it is not
 necessary.
 
 \end{itemize}
 
 To overcome these, we propose a general heuristic called {\it \nso} which  can be universally
 applied to a wide class of logics, often leading to inversion calculus.

\section{A Universal Heuristic for Theorem Proving}

Japaridze\cite{Jap03,Japtcs,Japic} have used an interesting heuristic in developing his proof theory, which
 is closely related to the Nongshim Cup game.
The Nongshim Cup is a team competition for the two Go-playing
countries such as China, Korea and Japan. It has an interesting 
tournament format: 
Two countries each starts with some number  (typically five each) of players. 
The winner of each subsequent game stays in to play the next opponent from another country.
In the end the team with remaining players  is the winner of the
tournament.

Let us assume Korea with $m$ korean players $\{ k_1,\ldots,k_m \}$ 
competes  against China with  $ n$ chinese players $\{ c_1,\ldots,c_n \}$.
This is written as \ns{Go}{$\{ k_1,\ldots,k_m \}$}{$\{ c_1,\ldots,c_n \}$}.
The following algorithm will clarify the Cup format.

\begin{description}

\item [step1:]  set current game to \ns{Go}{$\{ k_1 \}$}{$\{ c_1 \}$} 
\% initialization
\item [step2:] if Korea wins the current game of the form \ns{Go}{$\{ k_1,\ldots,k_i \}$}{$\{ c_1,\ldots,c_j \}$}, then     update   current game to \ns{Go}{$\{ k_1,\ldots,k_i \}$}{$\{ c_1,\ldots,c_j,c_{j+1} \}$} \\

      else      update current game to \ns{Go}{$\{ k_1,\ldots,k_i,k_{j+1} \}$}{$\{ c_1,\ldots,c_j \}$}

\item [step3:] repeat step 2 until one team has no remaining player.
      
  \end{description}
\noindent 
In this way,  the decision of winning/losing can be achieved at the
{\it earliest} possible time.

We are interested in applying the above idea to theorem proving.
 One instance of this heuristic -- $\ns{literalization(F)}{M}{O}$ --
has been used
to proof in first-order logic, 
where $F$ is a first-order logic formula, $M$ is a set of top-level occurrences of $\cle$-formulas
and $O$ is a set of top-level occurrences of
$\cla$-formulas. The process of literalization(F) transforms $F$ to its skeleton by
removing $M$ and $O$ in it, as we shall see below.
The resulting proof procedure\cite{Kwonheu17} turns out to be quite effective and performs better than the traditional inversion calculus.

In this paper, we apply this heuristic approach to propositional logic.
That is, we apply 

 \[ \ns{literalization(F)}{M}{O} \]
 
\noindent where 

\begin{itemize}

\item $F$ is a propositional  formula, 

\item $M$ is a set of top-level occurrences of $\mld$-formulas
and 

\item $O$ is a set of top-level occurrences of $\mlc$-formulas.

\end{itemize}

That is, the resulting procedure, which we call \cltw,  is a $game$-$viewed$ proof which 
captures  $game$-$playing$ nature  in proof search. It views 

\begin{enumerate}

\item sequents as games between the machine and the environment,

\item proofs as a winning strategy of the machine, and 

\item $\mlc$  as the env's resource and $\mld$ as the machine's resource.

\end{enumerate}

 For propositional logic, it turns out that this heuristic
 does not derive a new calculus.
 Instead, it is able to derive an existing invertible sequent calculus from sequent calculus.
This is meaningful, as  it provides a deeper insight into why 
all the rules are invertible in propositional logic\cite{Ketonen}.

\section{The  logic $\cltw$}\label{ss6}

 The  formulas  are   the standard  classical propositional formulas,
 with the features that
 (a) $\twg,\tlg$ are added, and (b) $\gneg$ is only allowed to be applied to atomic formulas.
Thus we assume that formulas are in negation normal form.

The  deductive system $\cltw$ below axiomatizes the set of valid propositional formulas. 
$\cltw$ is a one-sided sequent calculus system, where 
a  sequent is a multiset of formulas.
Our presentation  follows the one in \cite{Jap03}.

First, we need to define some terminology.

\begin{enumerate}
\item A {\bf surface occurrence} of a subformula is an occurrence that is 
not in the scope of any connectives ($\mlc$ and/or $\mld$). 
\item  A sequent is {\bf literal} iff all of its formulas are so. 
\item The {\bf literalization} $\elz{F}$ of a formula $F$ is the result of replacing
  in $F$ every surface occurrence of
  $\mld$-subformulas by $\tlg$, and every surface occurrence of
  $\mlc$-subformulas by $\twg$.
  
   The {\bf literalization} $\elz{F_1,\ldots,F_n}$ of a sequent 
$F_1,\ldots,F_n$ is the propositional formula $\elz{F_1}\mld\ldots\mld \elz{F_n}.$
\item A sequent  is said to be {\bf stable} iff its literalization is classically valid; otherwise it is {\bf unstable}.

\end{enumerate}

\begin{center}
\begin{picture}(100,30)

\put(0,10){\bf THE RULES OF $\cltw$}

\end{picture}
\end{center}

$\cltw$ has the four rules listed below where
 $\Gamma$ is a multiset of formulas and $F$ is a formula.

The  deductive system $\cltw$ is shown below. 
Below, $X$:stable means that $X$ must be  stable.
   Similarly  for $X$:unstable. The Fail rule reads:
   an unstable sequent $X$ containing no surface  occurrences of  $H_0 \mld H_1$ is not derivable.

  \begin{center}
\begin{picture}(74,70)
\put(12,50){\bf  Fail}
\put(20,30){$\tlg$}
\put(0,22){\line(1,0){45}}
\put(55,20){(no surface $H_0 \mld H_1$ in $X$)}
\put(8,8){$X$:unstable}
\end{picture}
\end{center}

\begin{center}
\begin{picture}(74,70)

\put(12,50){\bf $\mld$-rule}
\put(8,30){$\Gamma,F,G$}
\put(0,22){\line(1,0){78}}
\put(8,8){$\Gamma, F\mld G$:unstable}

\end{picture}
\end{center}

\begin{center}
\begin{picture}(70,70)
\put(12,50){\bf Succ}
\put(8,30){$\twg$}
\put(0,22){\line(1,0){45}}
\put(55,20){(no surface $H_0 \mlc H_1$ in $X$)}
\put(8,8){$X$: stable}
\end{picture}
\end{center}

\begin{center}
\begin{picture}(74,70)
\put(20,50){\bf $\mlc$-rule }
\put(20,30){$\Gamma,F$ \hspace{1em} $\Gamma,G$}
\put(0,22){\line(1,0){85}}
\put(20,8){$\Gamma, F\mlc G$: stable}
\end{picture}
\end{center}

A {\bf $\cltw$-proof} of a sequent $X$ is a sequence $X_1,\ldots,X_n$ of sequents, with $X_n=X$,
$X_1 = \twg$ such that, each $X_i$ follows  by one of the rules of $\cltw$ from $X_{i-1}$.


Below we describe some examples.

\begin{example}\label{j28b}
The formula $p(a)\mlc p(b), \neg p(a) \mld \neg p(b)$ is provable in $\cltw$ as follows:\vspace{7pt} \\\\
\noindent 1. $\begin{array}{l}
   p(b), \neg p(a),  \neg p(b)
\end{array}$  \ \ $Succ$\vspace{7pt}

\noindent 2. $\begin{array}{l}
   p(a), \neg p(a), \neg p(b)
\end{array}$  \ \ $Succ$\vspace{7pt}

\noindent 3. $\begin{array}{l}
   p(b), \neg p(a) \mld \neg p(b)
\end{array}$  \ \ $\mld$ from 1 \vspace{7pt}

\noindent 4. $\begin{array}{l}
   p(a), \neg p(a) \mld \neg p(b)
\end{array}$  \ \ $\mld$ from 2 \vspace{7pt}

\noindent 5. $\begin{array}{l}
   p(a)\mlc p(b), \neg p(a) \mld \neg p(b)
\end{array}$  \ \ $\mlc$ from 3,4 \vspace{7pt}

\end{example}

\section{The soundness and completeness of \cltw}\label{ssc}

We now present the soundness and completeness of \cltw.

\begin{theorem}\label{main}
  
  \begin{enumerate}

  \item If  $\cltw$ terminates with success for $X$, then $X$ is valid.

  \item If  $\cltw$ terminates with failure  for $X$, then $X$ is invalid.


    \end{enumerate}
\end{theorem}

\begin{proof}  Consider an arbitrary sequent $X$. 
\vspace{5pt}

{\em Soundness:}  Induction on the length of derivatons.

{\em Case 1:} $X$ is derived from $Y$ and $Z$ by $\mlc$-rule.
By the induction hypothesis,
both $Y$ and $Z$ are valid, which implies that  $X$ is valid.

{\em Case 2:} $X$ is derived from $Y$ by $\mld$-rule.
By the induction hypothesis,
 $Y$ is valid, which implies that  $X$ is valid.

{\em Case 3}:  $X$ is derived from $Y$ by Succ.

In this case, we know that there is no surface occurrences of
$F\mlc G$ in $X$ and 
  $\elz{X}$ is classically valid. It is then
  easy to see that, reversing the literalization
  of $\elz{X}$ (replacing $\tlg$ by any formula of the form
  $F\mld G$)
  preserves  validity. For example, if $X$ is $p\mli p, q\mld r$,
  then $\elz{X}$ is valid and $X$ is valid as well.

 \vspace{5pt}

{\em Completeness:}  Assume $\cltw$ terminates with failure. 
 
We proceed by induction on the length of derivations.

If $X$ is stable, then there should be a $\cltw$-unprovable sequent $Y$ with the following
condition.

 {\em Case 1: $\mlc$:} $X$ has the form $\Gamma,F\mlc G$, and
 $Y$ is either $\Gamma,F$ or $\Gamma,G$.
Suppose  $\Gamma,F$ is  $\cltw$-unprovable.
By the induction hypothesis, $\Gamma,F$ is invalid and $X$ is invalid
as well. Similarly for $\Gamma,G$.

Next, we consider the cases when $X$ is not stable. Then there are two
cases to consider.

{\em Case 2.1: Fail}: In this case,
there is no surface occurrence of $F\mld G$ and the algorithm terminates with failure.
As $X$ is not stable, $\elz{X}$ is not classically valid. If we reverse the propositionalization
of $\elz{X}$ by replacing $\twg$ by any formula with some surface occurrence of $F\mlc G$,
we observe that invalidity is preserved. Therefore, $X$ is not valid.

{\em Case 2.2: $\mld$}: In this case,
$X$ has the form $\Gamma,F\mld G$ and $Y$ is $\Gamma,F,G$.
In this case, $Y$ is a $\cltw$-unprovable sequent.
By the induction hypothesis,  $Y$ is  invalid. Therefore $X$ is not valid.

\end{proof}

\section{A simplified  $\cltw$}\label{sec:alg}

$\cltw$ in the previous section can be simplified by
observing the following: \\

A sequent $X$ is stable iff $X$ is either
$\Gamma,\top$ or $\Gamma,p,\gneg p$
or $\Gamma,F\mlc G$. \\

This observation leads to a simplified version of
$\cltw$, as shown below: \\

 \noindent{\bf Procedure} $pv(X)$:

 \begin{description}

\item[if  $X$ is  $\Gamma,p,\gneg p$] {\bf then}   return   Yes.

\item[elsif $X$ is $\Gamma,F\mlc G$] {\bf then}   return $pv(\Gamma,F)$ and $pv(\Gamma,G)$.

\item[elsif $X$ is $\Gamma,F\mld G$] {\bf then}   return $pv(\Gamma,F,G)$. 
  
\item[otherwise]  return No.
 
\end{description}\vspace{8pt}

Of course, we can speed up the above procedure by processing $\mld,\mlc$ in parallel. \\

 \noindent{\bf Procedure} $pv(X)$:

 \begin{description}

\item[if  $X$ is  $\Gamma,p,\gneg p$] {\bf then}   return   Yes.

\item[elsif $X$ is $\Gamma,F_1\mlc G_1,\ldots,F_n\mlc G_n$] {\bf then}   return \\ $pv(\Gamma,F_1,\ldots,F_n)$ and $\ldots$ and
  $pv(\Gamma,G_1,\ldots,G_n)$.\\
  \% total $2^n$ combinations above.

\item[elsif $X$ is $\Gamma,F_1\mld G_1,\ldots,F_n\mld G_n$] {\bf then}   return \\ $pv(\Gamma,F_1,G_1,\ldots,F_n,G_n)$. 
  
\item[otherwise]  return No.
 
\end{description}

\end{document}